\documentclass[11pt,amssymb,amsfont,a4paper]{article}
\usepackage{latexsym,amssymb,amsmath,amsthm,color}
\usepackage{geometry}
\usepackage{url}
\usepackage{graphicx}
\usepackage[utf8]{inputenc}
\usepackage[numbers]{natbib}
\usepackage{authblk}
\usepackage{tikz}

\theoremstyle{definition}
\newtheorem{definition}{Definition}

\theoremstyle{plain}
\newtheorem{theorem}{Theorem}
\newtheorem{proposition}[definition]{Proposition}
\newtheorem{lemma}[definition]{Lemma}
\newtheorem{remark}[definition]{Remark}
\newtheorem{corollary}[definition]{Corollary}

\DeclareMathOperator{\diag}{diag}
\DeclareMathOperator{\Row}{Row}
\DeclareMathOperator{\Col}{Col}
\DeclareMathOperator{\dd}{d}
\DeclareMathOperator{\wt}{wt}

\title{Hamming and simplex codes for the sum-rank metric}
\author[1]{Umberto Mart{\'i}nez-Pe\~{n}as \thanks{umberto@ece.utoronto.ca}}
\affil[1]{Dept.\ of Electrical \& Computer Engineering,
University of Toronto, Canada}

\date{}

\begin{document}

\maketitle

\begin{abstract}
Sum-rank Hamming codes are introduced in this work. They are essentially defined as the longest codes (thus of highest information rate) with minimum sum-rank distance at least $ 3 $ (thus one-error-correcting) for a fixed redundancy $ r $, base-field size $ q $ and field-extension degree $ m $ (i.e., number of matrix rows). General upper bounds on their code length, number of shots or sublengths and average sublength are obtained based on such parameters. When the field-extension degree is $ 1 $, it is shown that sum-rank isometry classes of sum-rank Hamming codes are in bijective correspondence with maximal-size partial spreads. In that case, it is also shown that sum-rank Hamming codes are perfect codes for the sum-rank metric. Also in that case, estimates on the parameters (lengths and number of shots) of sum-rank Hamming codes are given, together with an efficient syndrome decoding algorithm. Duals of sum-rank Hamming codes, called sum-rank simplex codes, are then introduced. Bounds on the minimum sum-rank distance of sum-rank simplex codes are given based on known bounds on the size of partial spreads. As applications, sum-rank Hamming codes are proposed for error correction in multishot matrix-multiplicative channels and to construct locally repairable codes over small fields, including binary.


\end{abstract}

\section{Introduction} \label{sec intro}

The Hamming metric and Hamming codes were both introduced by Hamming in the seminal work \cite{hamming}. In simple terms, Hamming codes can be understood as the longest linear codes that can correct a single symbol-wise error (thus having minimum Hamming distance at least $ 3 $) for a fixed redundancy (i.e., number of redundant symbols). Their binary version has been extensively applied in practice (usually interleaved or concatenated) for channels with rare bit-wise errors, since their decoding algorithms are very efficient and easy to implement even for large code lengths. 

Since then, other metrics have attracted the attention of coding theorists. Among them, the rank metric \cite{delsartebilinear,gabidulin} is a remarkable example due to its applications mainly in universal error-correction in linearly coded networks \cite{errors-network} and to construct locally repairable codes \cite{silberstein}. Traditionally, linear codes endowed with the rank metric are also called rank-metric codes.

However, in most applications, the length of a rank-metric code does not play the same role as if the Hamming metric were considered instead. The Hamming metric measures the number of symbol-wise errors in discrete memoryless channels (see \cite[Sec. 1.11.2]{pless}). Here, the code length represents the number of shots, corresponding to symbols, that the code covers in such a channel (a block of symbols, thus also the name block code for such codes). When one considers (in a discrete memoryless way) multiple uses or shots of the channels for which the rank metric measures the number of errors, one naturally obtains the so-called sum-rank metric. This metric was formally defined in \cite{multishot} and implicitly considered earlier for space-time coding in \cite{space-time-kumar}. It is the appropriate metric to measure the number of errors and erasures in multishot network coding universally \cite{secure-multishot}, obtain maximally recoverable locally repairable codes with arbitrary local codes \cite{lrc-sum-rank} and codes with optimal rate-diversity in space-time coding with multiple fading blocks \cite[Sec. III]{space-time-kumar}.

In this work, we introduce and briefly study Hamming codes tailored for the sum-rank metric, together with their duals, which we will call sum-rank simplex codes. We will define proper sum-rank Hamming codes (Definition \ref{def sum-rank hamming codes}), which are defined as linear codes capable of correcting a single sum-rank error for the largest number of shots (terms in the sum-rank sum) for a fixed base field, redundancy and matrix sizes. We will also define improper sum-rank Hamming codes (Definition \ref{def sum-rank hamming codes improper}) where the numbers of matrix columns are also allowed to grow. We will briefly comment on the longest rank-metric codes that may correct a single rank error, for fixed base field, redundancy and field-extension degree (or number of matrix rows). We will then turn to the case of field-extension degree $ 1 $ or $ 1 $ row per matrix. In this case, we will show that sum-rank isometry classes of sum-rank Hamming codes correspond bijectively with partial spreads for corresponding parameters. We will estimate the lengths, sublengths and number of shots for these sum-rank Hamming codes based on known bounds on the size of partial spreads. We will show that sum-rank Hamming codes are perfect codes for the sum-rank metric in this case, in contrast with the rank metric, for which there exist no non-trivial perfect codes. We will provide an efficient syndrome decoding algorithm for such codes, and we will lower bound the minimum sum-rank distance of their duals, called sum-rank simplex codes. 

The codes introduced in this work have applications in single error correction in multishot matrix-multiplicative channels (Subsection \ref{subsec matrix-mult channels}), which model for instance multishot network coding \cite{secure-multishot}. They also provide locally repairable codes over small fields, which can correct $ 2 $ extra erasures on top of the admissible local erasures and are backwards-compatible with arbitrary local linear codes, see Subsection \ref{subsec lrc}.

The remainder of this work is organized as follows. In Section \ref{sec sum-rank and isometries}, we define the sum-rank metric and characterize their linear isometries. In Section \ref{sec hamming codes}, we give the main definitions and properties of sum-rank Hamming codes. In Section \ref{sec sum-rank simplex codes}, we introduce and study their duals, that is, sum-rank simplex codes. Finally, in Section \ref{sec applications}, we provide the applications mentioned in the previous paragraph.

\section{The sum-rank metric and its linear isometries}  \label{sec sum-rank and isometries}

We start by fixing the notation that we will use throughout this paper. Fix a prime power $ q $ and positive integers $ m $, $ \ell $ and $ n = n_1 + n_2 + \cdots + n_\ell $. We will denote by $ \mathbb{F}_q $ the finite field with $ q $ elements. We will also denote by $ \mathbb{F}_q^{m \times n} $ the set of $ m \times n $ matrices with entries in $ \mathbb{F}_q $, and we denote $ \mathbb{F}_q^n = \mathbb{F}_q^{1 \times n} $. 

For matrices $ A_i \in \mathbb{F}_q^{n_i \times n_i} $, for $ i = 1,2, \ldots, \ell $, we define the \textit{block-diagonal matrix}
$$ \diag(A_1, A_2,\ldots, A_\ell) = \left( \begin{array}{cccc}
A_1 & 0 & \ldots & 0 \\
0 & A_2 & \ldots & 0 \\
\vdots & \vdots & \ddots & \vdots \\
0 & 0 & \ldots & A_\ell \\
\end{array} \right) \in \mathbb{F}_q^{n \times n}. $$

Fix an ordered basis $ \mathcal{A} = \{ \alpha_1, \alpha_2, \ldots, \alpha_m \} $ of $ \mathbb{F}_{q^m} $ over $ \mathbb{F}_q $. We define $ M_\mathcal{A} : \mathbb{F}_{q^m}^n \longrightarrow \mathbb{F}_q^{m \times n} $ as the matrix representation map, given by 
\begin{equation}
M_\mathcal{A} \left(  \mathbf{c}  \right) = \left( \begin{array}{cccc}
c_{11} & c_{12} & \ldots & c_{1n} \\
c_{21} & c_{22} & \ldots & c_{2n} \\
\vdots & \vdots & \ddots & \vdots \\
c_{m1} & c_{m2} & \ldots & c_{mn} \\
\end{array} \right),
\label{eq def matrix representation map}
\end{equation}
for $ \mathbf{c} = (c_1, c_2, \ldots, c_n) \in \mathbb{F}_{q^m}^n $, where $ c_{1,j}, c_{2,j}, \ldots, c_{m,j} \in \mathbb{F}_q $ are the unique scalars such that $ c_j = \sum_{i=1}^m \alpha_i c_{i,j} \in \mathbb{F}_{q^m} $, for $ j = 1,2, \ldots, n $. 

Given $ X \in \mathbb{F}_q^{m \times n} $, we denote by $ \Row(X) \subseteq \mathbb{F}_q^n $ and $ \Col(X) \subseteq \mathbb{F}_q^m $ the vector spaces generated by the rows and the columns of $ X $, respectively. For $ \mathbf{c} \in \mathbb{F}_{q^m}^n $, we denote $ \Row(\mathbf{c}) = \Row( M_\mathcal{A} (\mathbf{c}) ) \subseteq \mathbb{F}_q^n $ and $ \Col(\mathbf{c}) = \Col( M_\mathcal{A} (\mathbf{c}) ) \subseteq \mathbb{F}_q^m $. The latter depends on $ \mathcal{A} $, but we omit this for simplicity.

We now define the sum-rank metric in $ \mathbb{F}_{q^m}^n $, which was introduced in \cite{multishot}. It  was  implicitly  considered  earlier in the space-time coding literature (see \cite[Sec. III]{space-time-kumar}). The definition that we will consider was first given in \cite[Def. 25]{linearizedRS} and is more general, as the definition in \cite{multishot} considers equals sublengths $ n_1 = n_2 = \ldots = n_\ell $. The case of unequal sublengths is of interest, for instance, for locally repairable codes (see \cite{lrc-sum-rank}). 

\begin{definition} [\textbf{Sum-rank metric \cite{multishot}}]
Let $ \mathbf{c} = (\mathbf{c}^{(1)}, $ $ \mathbf{c}^{(2)}, $ $ \ldots, $ $ \mathbf{c}^{(\ell)}) \in \mathbb{F}_{q^m}^n $, where $ \mathbf{c}^{(i)} \in \mathbb{F}_{q^m}^{n_i} $, for $ i = 1,2, \ldots, \ell $. We define the sum-rank weight of $ \mathbf{c} $ as
$$ \wt_{SR}(\mathbf{c}) = \sum_{i=1}^\ell {\rm Rk}(M_{\mathcal{A}}(\mathbf{c}^{(i)})) , $$
where $ M_\mathcal{A} $ is as in (\ref{eq def matrix representation map}). Finally, we define the sum-rank metric $ \dd_{SR} : (\mathbb{F}_{q^m}^n)^2 \longrightarrow \mathbb{N} $ as $$ \dd_{SR}(\mathbf{c}, \mathbf{d}) = {\rm wt}_{SR} (\mathbf{c} - \mathbf{d}), $$ for all $ \mathbf{c}, \mathbf{d} \in \mathbb{F}_{q^m}^n $. In this work, a linear code is an $ \mathbb{F}_{q^m} $-linear vector subspace $ \mathcal{C} \subseteq \mathbb{F}_{q^m}^n $. We define its minimum sum-rank distance as
$$ {\rm d}_{SR}(\mathcal{C}) = \min \{ {\rm wt}_{SR}(\mathbf{c}) \mid \mathbf{c} \in \mathcal{C} \setminus \{ \mathbf{0} \} \}. $$
\end{definition}

Observe that indeed the sum-rank metric is a metric. Observe also that sum-rank weights and metrics depend on the sum-rank length partition $ n = n_1 + n_2 + \cdots + n_\ell $ and the subfield $ \mathbb{F}_q \subseteq \mathbb{F}_{q^m} $. However, we do not write this dependency for brevity. Note, on the other hand, that they do not depend on the choice of ordered basis $ \mathcal{A} $.  

Finally, we observe that the Hamming metric \cite{hamming} and the rank metric \cite{delsartebilinear,gabidulin} are recovered from the sum-rank metric by setting $ n_1 = n_2 = \ldots = n_\ell = 1 $ and $ \ell = 1 $, respectively. We will also use the notation $ {\rm wt}_R = {\rm wt}_{SR} $ and $ {\rm d}_R = {\rm d}_{SR} $ if $ \ell = 1 $, and $ {\rm wt}_H = {\rm wt}_{SR} $ and $ {\rm d}_H = {\rm d}_{SR} $ if $ n_1 = n_2 = \ldots = n_\ell = 1 $.

The next theorem was proven in \cite[Th. 1]{lrc-sum-rank}. It will be a central tool in the study of sum-rank Hamming codes.

\begin{theorem}[\textbf{\cite{lrc-sum-rank}}] \label{th sum-rank distance is min among hamming distances}
Given a linear code $ \mathcal{C} \subseteq \mathbb{F}_{q^m}^n $, it holds that
\begin{equation*}
\begin{split}
{\rm d}_{SR}(\mathcal{C}) = \min \{ {\rm d}_H(\mathcal{C} A) & \mid A = \diag(A_1, A_2, \ldots, A_\ell), \\
& A_i \in \mathbb{F}_q^{n_i \times n_i} \textrm{ invertible}, 1 \leq i \leq \ell \}.
\end{split}
\end{equation*}
\end{theorem}

In particular, we see that multiplying on the right by invertible block-diagonal matrices over $ \mathbb{F}_q $ constitutes a linear sum-rank isometry, for the corresponding sum-rank length partition. In \cite[Subsec. 3.3]{gsrws}, we anticipated a characterization of all linear sum-rank isometries. We now write such a characterization formally.

In this work, a map $ \phi : \mathbb{F}_{q^m}^n \longrightarrow \mathbb{F}_{q^m}^n $ is called linear if it is $ \mathbb{F}_{q^m} $-linear. As usual, we say that it is an isometry for a metric $ {\rm d} $ if 
$$ {\rm d}(\phi(\mathbf{c}), \phi(\mathbf{d})) = {\rm d}(\mathbf{c}, \mathbf{d}), $$
for all $ \mathbf{c}, \mathbf{d} \in \mathbb{F}_{q^m}^n $. Observe that a linear isometry is always injective (hence a vector space isomorphism if the domain and codomain have the same dimension) since $ {\rm d}(\mathbf{c}, \mathbf{d}) > 0 $ if, and only if, $ \mathbf{c} \neq \mathbf{d} $, therefore $ \phi(\mathbf{c}) \neq \phi(\mathbf{d}) $ if, and only if, $ \mathbf{c} \neq \mathbf{d} $. 

We may now prove the result mentioned above.

\begin{theorem} \label{th iso sum-rank}
Assume first that $ N = n_1 = n_2 = \ldots = n_\ell $ and $ n = \ell N $. A map $ \phi : \mathbb{F}_{q^m}^n \longrightarrow \mathbb{F}_{q^m}^n $ is a linear sum-rank isometry if, and only if, there exist elements $ \beta_1, \beta_2, \ldots, \beta_\ell \in \mathbb{F}_{q^m}^* $, invertible matrices $ A_1, A_2, \ldots, A_\ell \in \mathbb{F}_q^{N \times N} $ and a permutation $ \sigma : [\ell] \longrightarrow [\ell] $ such that
\begin{equation}
\phi(\mathbf{c}^{(1)}, \mathbf{c}^{(2)}, \ldots, \mathbf{c}^{(\ell)}) = (\beta_1 \mathbf{c}^{(\sigma(1))} A_1, \beta_2 \mathbf{c}^{(\sigma(2))} A_2, \ldots, \beta_\ell \mathbf{c}^{(\sigma(\ell))} A_\ell),
\label{eq iso equal lengths}
\end{equation}
for all $ \mathbf{c} = (\mathbf{c}^{(1)}, \mathbf{c}^{(2)}, \ldots, \mathbf{c}^{(\ell)}) \in \mathbb{F}_{q^m}^n $, where $ \mathbf{c}^{(i)} \in \mathbb{F}_{q^m}^{N} $, for $ i = 1,2, \ldots, \ell $.

More generally, assume that the sum-rank partition of $ n $ into sublengths is $ n = \ell_1 N_1 + \ell_2 N_2 + \cdots + \ell_v N_v $, where $ 0 < N_1 < N_2 < \ldots < N_v $ and $ \ell = \ell_1 + \ell_2 + \cdots + \ell_v $. That is, the first $ \ell_1 $ sublengths are equal, the next $ \ell_2 $ sublengths are equal and strictly larger than the previous ones, etcetera. In this case, a map $ \phi : \mathbb{F}_{q^m}^n \longrightarrow \mathbb{F}_{q^m}^n $ is a linear sum-rank isometry if, and only if, there exist linear sum-rank isometries $ \phi_j : \mathbb{F}_{q^m}^{\ell_j N_j} \longrightarrow \mathbb{F}_{q^m}^{\ell_j N_j} $ that are given as in (\ref{eq iso equal lengths}), for $ j = 1,2, \ldots, v $, such that
$$ \phi(\mathbf{c}) = (\phi_1(\mathbf{c}_1), \phi_2(\mathbf{c}_2), \ldots, \phi_v(\mathbf{c}_v)), $$
for all $ \mathbf{c} = (\mathbf{c}_1, \mathbf{c}_2, \ldots, \mathbf{c}_v) $, where $ \mathbf{c}_j \in \mathbb{F}_{q^m}^{\ell_j N_j} $, for $ j = 1,2, \ldots, v $.
\end{theorem}
\begin{proof}
In both the particular and general cases, the reversed implication is trivial, so we only prove the direct one. 

We first prove the particular case where $ N = n_1 = n_2 = \ldots = n_\ell $.  Reindex the canonical basis of $ \mathbb{F}_{q^m}^n $ as $ \mathbf{e}^{(i)}_j = \mathbf{e}_{(i-1)N + j} $, for $ i = 1,2, \ldots, \ell $ and $ j = 1,2,\ldots N $, and fix one such index $ i $. For simplicity, we identify $ \mathbf{e}^{(i)}_j \in \mathbb{F}_{q^m}^n $ with the $ j $th vector of the canonical basis in $ \mathbb{F}_{q^m}^N $, keeping in mind that we have fixed the index $ i $. Define now the linear map $ \psi_i : \mathbb{F}_{q^m}^N \longrightarrow \mathbb{F}_{q^m}^n $ such that
$$ \psi_i(\mathbf{c}^{(i)}) = \phi(\mathbf{0}, \ldots, \mathbf{0}, \mathbf{c}^{(i)}, \mathbf{0}, \ldots, \mathbf{0}), $$
where $ \mathbf{c}^{(i)} \in \mathbb{F}_{q^m}^N $ is placed in the $ i $th block of $ N $ coordinates in $ \mathbb{F}_{q^m}^n $. Obviously, we have that $ {\rm wt}_{SR}(\psi_i(\mathbf{c}^{(i)})) = {\rm wt}_R(\mathbf{c}^{(i)}) $. 

First, we see that $ \psi_i(\mathbf{e}_{j}^{(i)})^{(k)} = \mathbf{0} $ for all $ k \in [\ell] $ except for exactly one, which may depend on $ j $, for $ j = 1,2, \ldots, N $. Indeed if $ \psi_i(\mathbf{e}_{j}^{(i)})^{(k)} \neq \mathbf{0} $ for two $ k $, then $ {\rm wt}_{SR}(\psi_i(\mathbf{e}_{j}^{(i)})) \geq 2 $, which contradicts that $ {\rm wt}_R(\mathbf{e}_j^{(i)}) = 1 $.

Similarly we see that if $ \psi_i(\mathbf{e}_{j_1}^{(i)})^{(i_1)} \neq \mathbf{0} $, $ \psi_i(\mathbf{e}_{j_2}^{(i)})^{(i_2)} \neq \mathbf{0} $, and $ j_1 \neq j_2 $, then $ i_1 = i_2 $. If this does not hold, then by the previous paragraph, we have that
$$ 1 = {\rm wt}_R(\mathbf{e}_{j_1}^{(i)} + \mathbf{e}_{j_2}^{(i)}) = {\rm wt}_{SR}(\psi_i(\mathbf{e}_{j_1}^{(i)}) + \psi_i(\mathbf{e}_{j_2}^{(i)})) \geq 2, $$
which is absurd. Therefore for all $ i \in [\ell] $, there exists $ \sigma(i) \in [\ell] $ such that $ \psi_i(\mathbf{e}_{j}^{(i)})^{(\sigma(i))} \neq \mathbf{0} $ and $ \psi_i(\mathbf{e}_{j}^{(i)})^{(k)} = \mathbf{0} $ for $ k \neq \sigma(i) $. In other words, $ \pi_i \circ \psi_i : \mathbb{F}_{q^m}^N \longrightarrow \mathbb{F}_{q^m}^N $ is a linear rank isometry, where $ \pi_i : \mathbb{F}_{q^m}^n \longrightarrow \mathbb{F}_{q^m}^N $ is the projection onto the $ i $th block of $ N $ coordinates. By \cite[Th. 1]{berger}, there exists an element $ \beta_i \in \mathbb{F}_{q^m}^* $ and an invertible matrix $ A_i \in \mathbb{F}_q^{N \times N} $ such that $ \pi_i(\psi_i(\mathbf{c}^{(i)})) = \beta_i \mathbf{c}^{(i)}A_i $, for all $ \mathbf{c}^{(i)} \in \mathbb{F}_{q^m}^N $.

Finally, $ \sigma : [\ell] \longrightarrow [\ell] $ must be a bijection, i.e., a permutation, because otherwise $ \phi $ would not be a vector space isomorphism. Thus the result follows.

To extend this result to the general case (where the sublengths $ n_i $ may not be equal), we only need to prove that a block of $ n_i $ coordinates cannot be mapped to a block of $ n_k $ coordinates if $ n_i \neq n_k $. The proof of this claim is similar to the rest of this proof and is left to the reader.
\end{proof}

Finally, since MacWilliams' extension theorem does not hold in the case $ \ell = 1 $ (see \cite[Ex. 2.9(c)]{barra}), we choose to define sum-rank isometric codes (of the same length) as follows.

\begin{definition} \label{def sum-rank isometric codes}
We say that two linear codes $ \mathcal{C}, \mathcal{D} \subseteq \mathbb{F}_{q^m}^n $ are sum-rank isometric if there exists a linear sum-rank isometry $ \phi : \mathbb{F}_{q^m}^n \longrightarrow \mathbb{F}_{q^m}^n $ such that $ \mathcal{D} = \phi (\mathcal{C}) $.
\end{definition}

\section{Hamming codes for the sum-rank metric}  \label{sec hamming codes}

In this section, we introduce Hamming codes for the sum-rank metric, which we will call sum-rank Hamming codes for brevity.

\subsection{Definition and basic properties}  \label{subsec def and basic properties}

We start by giving the main definitions and basic properties. We will distinguish between codes with equal sublengths, which we will call proper, and those with unequal sublengths, which we will call improper.

\begin{definition} [\textbf{Proper sum-rank Hamming codes}] \label{def sum-rank hamming codes}
Fix the base field size $ q $, the extension degree $ m $ and a redundancy $ r $. Let $ n = \ell N $ be a sum-rank length partition with equal sublengths $ N = n_1 = n_2 = \ldots = n_\ell $. We say that a linear code $ \mathcal{C} \subseteq \mathbb{F}_{q^m}^n $ is a proper sum-rank Hamming code if $ {\rm d}_{SR}(\mathcal{C}) \geq 3 $, $ r = n - \dim(\mathcal{C}) $ and there is no other linear code $ \mathcal{C}^\prime \subseteq \mathbb{F}_{q^m}^{n^\prime} $ with $ {\rm d}_{SR}(\mathcal{C}^\prime) \geq 3 $ and $ r = n^\prime - \dim(\mathcal{C}^\prime) $ for a sum-rank length partition $ n^\prime = N \ell^\prime $ where $ \ell^\prime > \ell $.
\end{definition}

\begin{definition} [\textbf{Improper sum-rank Hamming codes}] \label{def sum-rank hamming codes improper}
Fix the base field size $ q $, the extension degree $ m $ and a redundancy $ r $. Let $ n = n_1 + n_2 + \cdots + n_\ell $ be a sum-rank length partition. We say that a linear code $ \mathcal{C} \subseteq \mathbb{F}_{q^m}^n $ is an improper sum-rank Hamming code if $ {\rm d}_{SR}(\mathcal{C}) \geq 3 $, $ r = n - \dim(\mathcal{C}) $ and there is no other linear code $ \mathcal{C}^\prime \subseteq \mathbb{F}_{q^m}^{n^\prime} $ with $ {\rm d}_{SR}(\mathcal{C}^\prime) \geq 3 $ and $ r = n^\prime - \dim(\mathcal{C}^\prime) $ for a sum-rank length partition $ n^\prime = n_1^\prime + n_2^\prime + \cdots + n^\prime_{\ell^\prime} $ where $ n^\prime > n $, $ \ell^\prime \geq \ell $ and $ n^\prime_i \geq n_i $, for $ i = 1,2, \ldots, \ell $.
\end{definition}

\begin{remark}
Classical Hamming codes \cite{hamming} (see also \cite[Sec. 1.8]{pless}) are therefore proper sum-rank Hamming codes where $ m = N = 1 $. To the best of our knowledge, the case $ \ell = 1 $, or Hamming codes for the rank metric, has not yet been considered. We will briefly treat it in Subsection \ref{subsec rank Hamming codes}.
\end{remark}

In other words, sum-rank Hamming codes are, by definition, the longest linear codes with minimum sum-rank distance at least $ 3 $ and redundancy $ r $, for a given base field and field-extension degree. The whole point of this definition is to obtain the codes with highest information rate for codes that can correct one sum-rank error, for a given base field and field-extension degree. The base field and field-extension degrees are fixed, since in general better codes (longer with equal or larger information rates) can be obtained over larger fields, establishing a natural trade-off, since larger fields imply higher computational complexity and less fine symbol partitions (meaning that errors on smaller symbols, such as bit errors, are treated as errors on larger symbols, such as bytes). 

In practice, one usually wishes to fix a code length and find the single error-correcting code with the largest information rate (i.e., largest dimension). However, as in the classical case, we will see that it is easier to fix the redundancy and try to find the longest single error-correcting code for that redundancy.

Due to Theorem \ref{th sum-rank distance is min among hamming distances}, it will be useful to use a basic characterization of codes with minimum Hamming distance at least $ 3 $. The following lemma is \cite[Cor. 1.4.14]{pless}.

\begin{lemma} \label{lemma min hamming distance 3}
Let $ \mathcal{C} \subseteq \mathbb{F}_{q^m}^n $ be a linear code with parity-check matrix $ H \in \mathbb{F}_{q^m}^{r \times n} $. Then $ {\rm d}_H(\mathcal{C}) \geq 3 $ if, and only if, any two columns of $ H $ are linearly independent over $ \mathbb{F}_{q^m} $.
\end{lemma}

Equipped with Theorem \ref{th sum-rank distance is min among hamming distances} and Lemma \ref{lemma min hamming distance 3}, we may show that Definitions \ref{def sum-rank hamming codes} and \ref{def sum-rank hamming codes improper} are consistent, meaning that both the length $ n $ and number of sublengths $ \ell $ of a sum-rank Hamming code (proper or improper) cannot be infinite. However, the next proposition does not imply that sum-rank Hamming codes exist for all the considered parameters. That is, if a code with minimum sum-rank distance at least $ 3 $ exists, for given base and extension fields, and given redundancy, then one such code with maximum possible length exists, which would then be a sum-rank Hamming code. We will actually give two bounds, using also \cite[Cor. 3]{lrc-sum-rank}, and give an upper bound on the average sublength. However, these bounds are not sharp in general, as shown later by Corollary \ref{cor length of proper sum-rank hamming m=1 }.

\begin{proposition}
Let $ \mathcal{C} \subseteq \mathbb{F}_{q^m}^n $ be a linear code such that $ {\rm d}_{SR}(\mathcal{C}) \geq 3 $ and $ n = n_1 + n_2 + \cdots + n_\ell $. Then it holds that
\begin{equation}
n \leq \min \left\lbrace \frac{q^{mr} - 1}{q^m - 1},  \left\lfloor \frac{\ell m r}{2} \right\rfloor  \right\rbrace ,
\label{eq upper bound length}
\end{equation}
where $ r = n - \dim(\mathcal{C}) $. In particular, we have the following upper bound on the average sublength:
\begin{equation}
\frac{\sum_{i=1}^\ell n_i}{\ell} \leq \frac{mr}{2}.
\label{eq upper bound average}
\end{equation}
In particular, for proper sum-rank Hamming codes with equal sublengths $ N = n_1 = n_2 = \ldots = n_\ell $, we deduce that $ 2N \leq mr $ and
\begin{equation}
\ell \leq \frac{q^{mr} - 1}{N(q^n - 1)}.
\label{eq upper bound for proper}
\end{equation}
\end{proposition}
\begin{proof}
From Theorem \ref{th sum-rank distance is min among hamming distances}, we deduce that $ {\rm d}_H(\mathcal{C}) \geq 3 $. Thus $ \mathcal{C} $ satisfies the conditions in Lemma \ref{lemma min hamming distance 3}. Hence its code length $ n $ is not larger than the size of the projective space $ \mathbb{P}_{\mathbb{F}_{q^m}}(\mathbb{F}_{q^m}^r) $, which is the first term in the minimum in (\ref{eq upper bound length}). 

On the other hand, the result \cite[Cor. 3]{lrc-sum-rank} says that 
$$ | \mathcal{C} | \leq \left( q^{n/\ell} \right)^{\ell m - {\rm d}_{SR}(\mathcal{C}) + 1}. $$
Using that $ | \mathcal{C} | = q^{m(n-r)} $, solving for $ {\rm d}_{SR} (\mathcal{C}) $ and using that $ {\rm d}_{SR}(\mathcal{C}) \geq 3 $, we deduce that
$$ 3 \leq {\rm d}_{SR}(\mathcal{C}) \leq \frac{\ell m r}{n} + 1, $$
from which the second term in the minimum in (\ref{eq upper bound length}) follows, and we are done.
\end{proof}

\begin{remark}
Note that, for classical Hamming codes, we have that $ N = n_1 = n_2 = \ldots = n_\ell = 1 $ and $ m = 1 $, thus $ n = \ell $ and (\ref{eq upper bound length}) simply says that $ r \geq 2 $ (and actually any $ r \geq 2 $ can be chosen), which is obvious from the Singleton bound, and $ n \leq (q^r-1)/(q-1) $, which is the actual length of classical Hamming codes (see \cite[Sec. 1.8]{pless}).
\end{remark}

We next give a basic characterization of having minimum sum-rank distance at least $ 3 $. The proof is immediate from Theorem \ref{th sum-rank distance is min among hamming distances} and Lemma \ref{lemma min hamming distance 3} after unfolding the definitions. We will make use of this result mainly for the case $ m = 1 $ in Subsection \ref{subsec partial spreads}.

\begin{proposition} \label{prop sr distance at least 3}
Let $ n = n_1 + n_2 + \cdots + n_\ell $ be a sum-rank length partition. Let $ \mathcal{C} \subseteq \mathbb{F}_{q^m}^n $ be a linear code with a parity check matrix $ H = (H_1, H_2, \ldots, H_\ell) \in \mathbb{F}_{q^m}^{r \times n} $, where $ H_i \in \mathbb{F}_{q^m}^{r \times n_i} $, for $ i = 1,2, \ldots, \ell $. Then it holds that $ {\rm d}_{SR}(\mathcal{C}) \geq 3 $ if, and only if, the following two conditions hold:
\begin{enumerate}
\item
$ H_i \mathbf{a}_i^T \in \mathbb{F}_{q^m}^r $ and $ H_j \mathbf{a}_j^T \in \mathbb{F}_{q^m}^r $ are linearly independent over $ \mathbb{F}_{q^m} $, for all non-zero $ \mathbf{a}_i \in \mathbb{F}_q^{n_i} $ and $ \mathbf{a}_j \in \mathbb{F}_q^{n_j} $, and for all $ 1 \leq i < j \leq \ell $, and
\item
$ H_i \mathbf{a}^T \in \mathbb{F}_{q^m}^r $ and $ H_i \mathbf{b}^T \in \mathbb{F}_{q^m}^r $ are linearly independent over $ \mathbb{F}_{q^m} $, for all $ \mathbf{a} \in \mathbb{F}_q^{n_i} $ and $ \mathbf{b} \in \mathbb{F}_q^{n_i} $ that are linearly independent over $ \mathbb{F}_q $, for all $ i = 1,2, \ldots, \ell $.
\end{enumerate}
\end{proposition}

\subsection{On the case $ \ell = 1 $: Rank-metric Hamming codes} \label{subsec rank Hamming codes}

In this subsection, we briefly discuss the rank-metric case, that is, $ \ell = 1 $. Observe that, in principle, there may not exist sum-rank Hamming codes with $ \ell = 1 $, whether proper or improper. This is because, both in Definitions \ref{def sum-rank hamming codes} and \ref{def sum-rank hamming codes improper}, we do not fix $ \ell $ and we allow it to grow. In other words, it may be the case that, for a given redundancy, base field and extension field, there always exists a code with minimum sum-rank distance at least $ 3 $ and $ \ell > 1 $. Thus, for such redundancy, base and extension fields, no rank-metric Hamming code would exist. We believe that looking for the largest number of shots $ \ell $ is the right way of considering Hamming codes, as $ \ell $ is the dominant part of the code length and has analogous interpretations to the classical code length in the applications. See the Introduction or Section \ref{sec applications} for details.

However, we briefly discuss what happens if we decide to fix $ \ell = 1 $ and find the longest linear rank-metric code $ \mathcal{C} \subseteq \mathbb{F}_{q^m}^n $ with $ {\rm d}_R(\mathcal{C}) \geq 3 $ for a fixed redundancy $ r = n - \dim(\mathcal{C}) $, base field $ q $ and field-extension degree $ m $. Whether such codes are sum-rank Hamming codes according to Definitions \ref{def sum-rank hamming codes} and \ref{def sum-rank hamming codes improper} is left open. 

First, it is not difficult to see that the minimum in (\ref{eq upper bound length}) is attained by the second term. Hence we have the bound
\begin{equation}
2n \leq rm.
\label{eq upper bound for rank hamming}
\end{equation}

The bound (\ref{eq upper bound for rank hamming}) can be attained for many choices of parameters, although we leave it as open problem to see if it can always be achieved as long as $ 2 $ divides $ rm $.

The following constructions from the literature based on Gabidulin codes \cite{gabidulin} achieve the bound (\ref{eq upper bound for rank hamming}). Assume that $ 2n = rm $. First note that the case $ r = 1 $ may not happen, because then it would hold that $ 2n = m $ by (\ref{eq upper bound for rank hamming}) and $ 3 \leq {\rm d}_R(\mathcal{C}) \leq r + 1 = 2 $ by the Singleton bound. Hence $ r \geq 2 $ and therefore, $ n \geq m $. 

\begin{enumerate}
\item
If $ r = 2 $, then $ n = m $. In this case, the Gabidulin code \cite{gabidulin} of length $ n $ and dimension $ n - 2 $ has minimum rank distance exactly $ {\rm d}_R(\mathcal{C}) = n - \dim(\mathcal{C}) + 1 = r + 1 = 3 $.
\item
In general for $ r $ even, let $ r = 2 h $ for a positive integer $ h $. Then a Cartesian product of $ h $ Gabidulin codes, each of length $ m $ and dimension $ m-2 $, gives a linear code $ \mathcal{C} \subseteq \mathbb{F}_{q^m}^n $ of length $ n = hm $, dimension $ k = h(m-2) = n - r $ and minimum rank distance $ {\rm d}_R(\mathcal{C}) = 3 $. It is easy to check that $ 2n = mr $, hence (\ref{eq upper bound for rank hamming}) is attained. 
\end{enumerate}

Decoding such codes can be done using decoding algorithms for Gabidulin codes. Unfortunately, as is well-known, decoding one Gabidulin code of length $ n $ is not efficient in practice, as the known decoding algorithms require in general a super-linear (in $ n $) number of operations over $ \mathbb{F}_{q^m} $, whose size is exponential in the code length $ n $. The sum-rank Hamming codes with $ m = 1 $ from Subsection \ref{subsec partial spreads}, in contrast, admit very efficient decoding algorithms (see Subsection \ref{subsec syndrome decoding m=1 }).

Finally, we observe that there exists a generalization of Gabidulin codes for the sum-rank metric in general. Such codes are called linearized Reed-Solomon codes and were introduced in general in \cite{linearizedRS}. Such codes recover classical Reed-Solomon codes when $ n_1 = n_2 = \ldots = n_\ell = 1 $, that is, when the sum-rank metric becomes the Hamming metric. For this reason, and since Reed-Solomon codes and Hamming codes are not equal (for the Hamming metric), we do not expect linearized Reed-Solomon codes to be sum-rank Hamming codes in general.

\subsection{The case $ m = 1 $: Maximal-size partial spreads}  \label{subsec partial spreads}

In this subsection, we study the case $ m = 1 $. Although this case may seem mathematically simpler, it is arguably the most interesting case from a practical point of view, as the size of the field $ q^m $ becomes simply $ q $. As we shall see, sum-rank isometry classes of sum-rank Hamming codes are in correspondence with maximal-size partial spreads \cite{beutel}.

We observe that, when $ m = 1 $, the sum-rank metric becomes the Hamming metric over the alphabets $ \mathbb{F}_q^{n_1}, \mathbb{F}_q^{n_2}, \ldots, \mathbb{F}_q^{n_\ell} $. If not all sublengths $ n_1, n_2, \ldots, n_\ell $ are equal, the corresponding codes are called mixed codes or polyalphabetic codes in the literature \cite{Etzion-constructions,Herzog,Sidorenko-polyalphabetic}. Constructions of Hamming codes for a single alphabet of the form $ \mathbb{F}_q^N $ (that is, the case $ N = n_1 = n_2 = \ldots = n_\ell $), with $ N > 1 $, were also given in \cite{etzion-byte}. However, the results in this subsection were not obtained in previous works, e.g., \cite{etzion-byte,Etzion-constructions,Herzog,Sidorenko-polyalphabetic}, to the best of our knowledge. 

We start by revisiting the definition of partial spreads.

\begin{definition}[\textbf{Partial spreads}]  \label{def partial spread}
Let $ 1 \leq N \leq r $. A partial $ N $-spread in $ \mathbb{F}_q^r $ is a family $ \mathcal{P} = \{ \mathcal{H}_i \}_{i \in I} $ of $ N $-dimensional vector subspaces $ \mathcal{H}_i \subseteq \mathbb{F}_q^r $, for $ i \in I $, such that $ \mathcal{H}_i \cap \mathcal{H}_j = \{ \mathbf{0} \} $, whenever $ i \neq j $. A maximal-size partial $ N $-spread is a partial $ N $-spread $ \mathcal{P} = \{ \mathcal{H}_i \}_{i \in I} $ such that $ I $ has maximal size among partial $ N $-spreads in $ \mathbb{F}_q^r $.
\end{definition}

We next characterize proper sum-rank Hamming codes when $ m = 1 $ in terms of partial spreads. 

\begin{theorem} \label{th SR Hamming given by partial spreads}
Let $ n = \ell N $ be a sum-rank length partition with equal sublengths $ N = n_1 = n_2 = \ldots = n_\ell $. A linear code $ \mathcal{C} \subseteq \mathbb{F}_q^n $ is a proper sum-rank Hamming code if, and only if, $ N \leq r = n - \dim(\mathcal{C}) $ and it has a parity-check matrix of the form
\begin{equation}
H = (H_1, H_2, \ldots, H_\ell) \in \mathbb{F}_q^{r \times n},
\label{eq parity-check SR hamming proper partial spreads}
\end{equation} 
where $ H_i \in \mathbb{F}_q^{r \times N} $ and $ \mathcal{H}_i = {\rm Col}(H_i) \subseteq \mathbb{F}_q^r $, for $ i = 1,2, \ldots, \ell $, and $ \mathcal{P} = \{ \mathcal{H}_1, \mathcal{H}_2, \ldots, \mathcal{H}_\ell \} $ is a maximal-size partial $ N $-spread in $ \mathbb{F}_q^r $.
\end{theorem}
\begin{proof}
First, Item 2 in Proposition \ref{prop sr distance at least 3} shows that all $ \mathcal{H}_i \subseteq \mathbb{F}_q^r $ must have dimension $ N $ and, in particular, $ N \leq r $. Second, Item 1 in Proposition \ref{prop sr distance at least 3} shows that it must hold that $ \mathcal{H}_i \cap \mathcal{H}_j = \{ \mathbf{0} \} $, whenever $ i \neq j $. Finally, $ \ell $ is maximum as in Definition \ref{def sum-rank hamming codes} if, and only if, $ \mathcal{P} $ is a maximal-size partial $ N $-spread, since  $ \ell = |\mathcal{P}| $.
\end{proof}

Now we show that sum-rank isometry classes of sum-rank Hamming codes are in bijective correspondence with maximal-size partial spreads.

\begin{theorem} \label{th SR Hamming corresp partial spreads}
Let $ n = \ell N $ be a sum-rank length partition with equal sublengths $ N = n_1 = n_2 = \ldots = n_\ell $. Let $ \mathcal{C}, \mathcal{D} \subseteq \mathbb{F}_q^n $ be two proper sum-rank Hamming codes with $ r = n - \dim(\mathcal{C}) = n - \dim(\mathcal{D}) $ and with parity-check matrices given by
$$ H = (H_1, H_2, \ldots, H_\ell) \in \mathbb{F}_q^{r \times n} \quad \textrm{and} \quad L = (L_1, L_2, \ldots, L_\ell) \in \mathbb{F}_q^{r \times n}, $$
respectively. Denote $ \mathcal{H}_i = {\rm Col}(H_i) \subseteq \mathbb{F}_q^r $ and $ \mathcal{L}_i = {\rm Col}(L_i) \subseteq \mathbb{F}_q^r $, for $ i = 1,2, \ldots, \ell $. Then $ \mathcal{C} $ and $ \mathcal{D} $ are sum-rank isometric if, and only if, 
$$ \{ \mathcal{H}_1, \mathcal{H}_2, \ldots, \mathcal{H}_\ell \} = \{ \mathcal{L}_1, \mathcal{L}_2, \ldots, \mathcal{L}_\ell \}. $$
In particular, sum-rank isometry classes of proper sum-rank Hamming codes with the sum-rank length partition $ n = \ell N $ and redundancy $ r $ correspond bijectively with maximal-size partial $ N $-spreads in $ \mathbb{F}_q^r $.
\end{theorem}
\begin{proof}
It follows by combining Definition \ref{def sum-rank isometric codes}, Theorem \ref{th iso sum-rank} and Theorem \ref{th SR Hamming given by partial spreads}.
\end{proof}

\begin{remark}
Observe that, by choosing $ N = 1 $, we recover the correspondence between classical Hamming codes \cite{hamming} and projective spaces, which obviously are maximal-size partial $ 1 $-spreads. See \cite[Sec. 1.8]{pless}.
\end{remark}

Next we recall some exact formulas and explicit bounds for the size of maximal-size proper partial spreads. We will use these results to estimate the length of sum-rank Hamming codes and the minimum sum-rank distance of sum-rank Simplex codes in Section \ref{sec sum-rank simplex codes}. There are numerous works in that respect. However, we will use the following two bounds from the literature due to their simplicity.

\begin{proposition} \label{prop size of maximal partial spread}
Let $ 1 \leq N \leq r $, and let $ s \geq 0 $ be the remainder of $ r $ divided by $ N $. If $ \mathcal{P} = \{ \mathcal{H}_i \}_{i \in I} $ is a maximal-size partial $ N $-spread in $ \mathbb{F}_q^r $, then 
\begin{equation}
\frac{q^r - q^s}{q^N - 1} - q^s + 1 \leq | I | \leq \frac{q^r - q^s}{q^N - 1} .
\label{eq inequalities max partial spreads}
\end{equation}
In particular, if $ s = 0 $, that is, if $ N $ divides $ r $, then $ \mathcal{P} $ is a maximal-size partial $ N $-spread if, and only if,
$$ | I | = \frac{q^r - 1}{q^N - 1} . $$
\end{proposition}
\begin{proof}
The lower bound in (\ref{eq inequalities max partial spreads}) is proven in \cite{beutel2}. A simple proof of the upper bound in (\ref{eq inequalities max partial spreads}) is given in \cite{spreads-network}, although sharper bounds were known earlier (see, e.g., \cite{beutel}).
\end{proof}

Therefore, we may estimate the length of proper sum-rank Hamming codes when $ m = 1 $ as follows.

\begin{corollary} \label{cor length of proper sum-rank hamming m=1 }
Let $ n = \ell N $ be a sum-rank length partition with equal sublengths $ N = n_1 = n_2 = \ldots = n_\ell $. Let $ r \geq N $, and let $ s \geq 0 $ be the remainder of $ r $ divided by $ N $. If $ \mathcal{C} \subseteq \mathbb{F}_q^n $ is a proper sum-rank Hamming code with $ r = n - \dim(\mathcal{C}) $, then 
\begin{equation}
N \left( \frac{q^r - q^s}{q^N - 1} - q^s + 1 \right) \leq n \leq N \frac{q^r - q^s}{q^N - 1} ,
\label{eq length proper sum-rank hamming m=1}
\end{equation}
and equality holds in (\ref{eq length proper sum-rank hamming m=1}) if $ s = 0 $.
\end{corollary}

\begin{remark}
Observe that the length $ n $ of a proper sum-rank Hamming code in (\ref{eq length proper sum-rank hamming m=1}) coincides with that of classical Hamming codes if $ N = m = 1 $, which implies that $ s = 0 $. See \cite[Sec. 1.8]{pless}.
\end{remark}

Finally, we show that proper sum-rank Hamming codes for $ m = 1 $ are perfect codes for such a sum-rank metric when $ N $ divides $ r $. In particular, we conclude that their sum-rank error-correction capability is exactly $ 1 $. 

\begin{corollary} \label{cor perfect code}
Let $ n = \ell N $ be a sum-rank length partition with equal sublengths $ N = n_1 = n_2 = \ldots = n_\ell $. Let $ r \geq N $ be such that $ N $ divides $ r $. If $ \mathcal{C} \subseteq \mathbb{F}_q^n $ is a proper sum-rank Hamming code with $ r = n - \dim(\mathcal{C}) $, then $ \mathcal{C} $ is a perfect code for such a sum-rank metric. In other words, if $ \mathcal{B}_t = \{ \mathbf{c} \in \mathbb{F}_q^n \mid {\rm wt}_{SR}(\mathbf{c}) \leq t \} $ denotes the ball centered in the origin with radius $ t > 0 $, then
$$ | \mathcal{C} | \cdot | \mathcal{B}_1 | = | \mathbb{F}_q^n |, $$
where $ 1 = \left\lfloor ({\rm d}_{SR}(\mathcal{C}) - 1)/2 \right\rfloor $ is the sum-rank error-correcting capability of $ \mathcal{C} $.
\end{corollary}
\begin{proof}
It is straightforward to see that
$$ | \mathcal{B}_1 | = 1 + \ell (q^N - 1). $$
Next, from Corollary \ref{cor length of proper sum-rank hamming m=1 } it holds that
$$ \ell = \frac{q^r - 1}{q^N - 1}. $$
Therefore, if we denote $ k = \dim(\mathcal{C}) $, we conclude that
$$ | \mathcal{C} | \cdot | \mathcal{B}_1 | = q^k (1 + \ell (q^N - 1)) = q^k \left( 1 + \frac{q^r - 1}{q^N - 1} (q^N - 1) \right) = q^k q^r = q^n = | \mathbb{F}_q^n |. $$
Finally, it must hold that $ t = \left\lfloor ({\rm d}_{SR}(\mathcal{C}) - 1 ) / 2 \right\rfloor = 1 $, since $ t \geq 1 $ and $ | \mathcal{C} | \cdot | \mathcal{B}_t | \leq | \mathbb{F}_q^n | = | \mathcal{C} | \cdot | \mathcal{B}_1 | $. 
\end{proof}

\begin{remark}
In \cite{loidreau}, it was proven that there exist no non-trivial perfect code for the rank-metric, that is, for the sum-rank metric where $ \ell = 1 $. Note however that Corollary \ref{cor perfect code} holds only for $ \ell \geq {\rm d}_{SR}(\mathcal{C}) = 3 $ since $ m = 1 $, hence it has no intersection with the case $ \ell = 1 $. However, the case of classical Hamming codes also falls into the case $ m = 1 $. Therefore, it is still an open problem to see if there exist non-trivial perfect codes for the sum-rank metric when $ m > 1 $. 
\end{remark}

Partial spreads allowed us to deal with proper Hamming codes. For improper Hamming codes, we need to define partial spreads whose subspaces may have distinct dimensions.

\begin{definition} [\textbf{Improper partial spreads}] \label{def partial spread improper}
Let $ \ell $ and $ n_1, n_2, \ldots, n_\ell $ be positive integers such that $ n_i \leq r $, for $ i = 1,2, \ldots, \ell $, and denote $ \mathbf{n} = (n_1, n_2, \ldots, n_\ell) $. An improper partial $ \mathbf{n} $-spread in $ \mathbb{F}_q^r $ is a family $ \mathcal{P} = \{ \mathcal{H}_i \}_{i = 1}^\ell $ of vector subspaces $ \mathcal{H}_i \subseteq \mathbb{F}_q^r $, for $ i = 1,2, \ldots, \ell $, such that $ \dim(\mathcal{H}_i) = n_i $, for $ i = 1,2, \ldots, \ell $, and $ \mathcal{H}_i \cap \mathcal{H}_j = \{ \mathbf{0} \} $, whenever $ i \neq j $. A maximal-size improper partial $ \mathbf{n} $-spread in $ \mathbb{F}_q^r $ is an improper partial $ \mathbf{n} $-spread in $ \mathbb{F}_q^r $ such that there is no improper partial $ \mathbf{n}^\prime $-spread in $ \mathbb{F}_q^r $ for $ \mathbf{n}^\prime = (n_1^\prime, n_2^\prime, \ldots, n_{\ell^\prime}^\prime) $, where $ \ell^\prime \geq \ell $, $ n_i^\prime \geq n_i $, for $ i = 1,2, \ldots, \ell $, and $ n_1^\prime + n_2^\prime + \cdots + n_{\ell^\prime}^\prime > n_1 + n_2 + \cdots + n_\ell $.
\end{definition}

All the previous results on proper sum-rank Hamming codes still hold for improper sum-rank Hamming codes with $ m = 1 $, except for being perfect codes, which we leave as open problem. We also leave as an open problem estimating their code length $ n $.

Next we collect the analogous results to Theorems \ref{th SR Hamming given by partial spreads} and \ref{th SR Hamming corresp partial spreads}. Proofs are left to the reader.

\begin{theorem} \label{th SR Hamming given by partial spreads improper}
Let $ \ell $ and $ n_1, n_2, \ldots, n_\ell $ be positive integers such that $ n_i \leq r $, for $ i = 1,2, \ldots, \ell $, and denote $ \mathbf{n} = (n_1, n_2, \ldots, n_\ell) $ and $ n = n_1 + n_2 + \cdots + n_\ell $. A linear code $ \mathcal{C} \subseteq \mathbb{F}_q^n $ is an improper sum-rank Hamming code for such a sum-rank length partition if, and only if, $ \max \{ n_1, n_2, \ldots, n_\ell \} \leq r = n - \dim(\mathcal{C}) $ and it has a parity-check matrix of the form
\begin{equation}
H = (H_1, H_2, \ldots, H_\ell) \in \mathbb{F}_q^{r \times n},
\label{eq parity-check SR hamming proper partial spreads improper}
\end{equation} 
where $ H_i \in \mathbb{F}_q^{r \times n_i} $ and $ \mathcal{H}_i = {\rm Col}(H_i) \subseteq \mathbb{F}_q^r $, for $ i = 1,2, \ldots, \ell $, and $ \mathcal{P} = \{ \mathcal{H}_1, \mathcal{H}_2, \ldots, \mathcal{H}_\ell \} $ is a maximal-size improper partial $ \mathbf{n} $-spread in $ \mathbb{F}_q^r $.
\end{theorem}

\begin{theorem} \label{th SR Hamming corresp partial spreads improper}
Let $ \ell $ and $ n_1, n_2, \ldots, n_\ell $ be positive integers such that $ n_i \leq r $, for $ i = 1,2, \ldots, \ell $, and denote $ \mathbf{n} = (n_1, n_2, \ldots, n_\ell) $ and $ n = n_1 + n_2 + \cdots + n_\ell $. Let $ \mathcal{C}, \mathcal{D} \subseteq \mathbb{F}_q^n $ be two improper sum-rank Hamming codes for such a sum-rank length partition and with $ r = n - \dim(\mathcal{C}) = n - \dim(\mathcal{D}) $ and parity-check matrices given by
$$ H = (H_1, H_2, \ldots, H_\ell) \in \mathbb{F}_q^{r \times n} \quad \textrm{and} \quad L = (L_1, L_2, \ldots, L_\ell) \in \mathbb{F}_q^{r \times n}, $$
respectively. Denote $ \mathcal{H}_i = {\rm Col}(H_i) \subseteq \mathbb{F}_q^r $ and $ \mathcal{L}_i = {\rm Col}(L_i) \subseteq \mathbb{F}_q^r $, for $ i = 1,2, \ldots, \ell $. Then $ \mathcal{C} $ and $ \mathcal{D} $ are sum-rank isometric if, and only if, 
$$ \{ \mathcal{H}_1, \mathcal{H}_2, \ldots, \mathcal{H}_\ell \} = \{ \mathcal{L}_1, \mathcal{L}_2, \ldots, \mathcal{L}_\ell \}. $$
In particular, sum-rank isometry classes of improper sum-rank Hamming codes with the sum-rank length partition $ n = n_1 + n_2 + \cdots + n_\ell $ and redundancy $ r $ correspond bijectively with maximal-size improper partial $ \mathbf{n} $-spreads in $ \mathbb{F}_q^r $.
\end{theorem}

\subsection{Syndrome decoding when $ m = 1 $} \label{subsec syndrome decoding m=1 }

In this subsection, we show how to adapt the single-error syndrome decoding algorithm for classical Hamming codes (see \cite[Sec. 1.11]{pless}) to sum-rank Hamming codes when $ m = 1 $, both proper and improper. Actually, the algorithm works exactly in the same way for any linear code of sum-rank distance at least $ 3 $ (thus being able to correct at least a single sum-rank error) whenever $ m = 1 $.

Set $ m = 1 $ and let $ n = n_1 + n_2 + \cdots + n_\ell $ be a sum-rank length partition. Consider in general a linear code $ \mathcal{C} \subseteq \mathbb{F}_q^n $ such that $ {\rm d}_{SR}(\mathcal{C}) \geq 3 $. As in the proof of Theorem \ref{th SR Hamming given by partial spreads}, any parity-check matrix of $ \mathcal{C} $ is of the form
$$ H = (H_1, H_2, \ldots, H_\ell) \in \mathbb{F}_q^{r \times n}, $$
where $ r = n - \dim(\mathcal{C}) $, and if $ \mathcal{H}_i = {\rm Col}(H_i) \subseteq \mathbb{F}_q^r $, for $ i = 1,2, \ldots, \ell $, then $ \mathcal{P} = \{ \mathcal{H}_1, \mathcal{H}_2, \ldots, \mathcal{H}_\ell \} $ forms a (not necessarily maximal-size) improper partial $ \mathbf{n} $-spread, where $ \mathbf{n} = (n_1, n_2, \ldots, n_\ell) $. In particular, the rank of $ H_i $ is $ n_i $, for $ i = 1,2, \ldots, \ell $, and thus $ r \geq \max \{ n_1, n_2, \ldots, n_\ell \} $.

We briefly recall how syndrome decoding works in general (for any metric). We refer to \cite[Sec. 1.11]{pless} for more details. Let $ t $ be a positive integer such that $ 2t $ is less than the minimum distance of $ \mathcal{C} $. Then there exists a unique vector (called coset leader) of weight at most $ t $ in every coset $ \mathbf{x} + \mathcal{C} $, where $ \mathbf{x} \in \mathbb{F}_q^n $, that contains a vector of weight at most $ t $. Recall that, for any $ \mathbf{x} \in \mathbb{F}_q^n $, it holds that $ H \mathbf{y}^T = H \mathbf{x}^T $, for all $ \mathbf{y} \in \mathbf{x} + \mathcal{C} $. Thus, we may identify coset leaders $ \mathbf{e} $ of weight at most $ t $ with syndromes $ \mathbf{s} = H\mathbf{x}^T $ of cosets $ \mathbf{x} + \mathcal{C} $ of weight at most $ t $.

With this in mind, the general algorithm works as follows: Let $ \mathbf{y} = \mathbf{c} + \mathbf{e} \in \mathbb{F}_q^n $ be the received vector, where $ \mathbf{c} \in \mathcal{C} $ is the correct codeword and $ \mathbf{e} \in \mathbb{F}_q^n $ is an error vector of weight at most $ t $. Compute the syndrome $ \mathbf{s} = H \mathbf{y}^T $. Then $ \mathbf{e} $ is the unique coset leader of the coset corresponding to $ \mathbf{s} $. If we have a method (e.g., by a precomputed lookup table) to obtain $ \mathbf{e} $ from $ \mathbf{s} $, then we obtain $ \mathbf{e} $ and compute $ \mathbf{c} = \mathbf{y} - \mathbf{e} $.

Thus the difficult part is to obtain the coset leader $ \mathbf{e} $ from the syndrome $ \mathbf{s} $. This is simple for the code $ \mathcal{C} $ described in the beginning of this subsection. Since we target $ t = 1 $, we only need to consider $ \mathbf{e} \in \mathbb{F}_q^n $ such that $ {\rm wt}_{SR}(\mathbf{e}) \leq 1 $. These are either $ \mathbf{e} = \mathbf{0} $ or
\begin{equation}
\mathbf{e}_{i,\mathbf{a}} = (\mathbf{0}, \ldots, \mathbf{0}, \mathbf{a}, \mathbf{0}, \ldots, \mathbf{0}) \in \mathbb{F}_q^n,
\label{eq def of coset leader for syndrome}
\end{equation}
where $ \mathbf{a} \in \mathbb{F}_q^{n_i} $ is a non-zero vector in the $ i $th block of $ n_i $ coordinates over $ \mathbb{F}_q $. Therefore, we only need to find $ i $ and $ \mathbf{a} $ from the syndrome $ H \mathbf{e}_{i,\mathbf{a}} $. To that end, we only need to try to solve the systems of linear equations
$$ H_j \mathbf{a}^T = \mathbf{s}, $$
given the syndrome $ \mathbf{s} \in \mathbb{F}_q^r $, iteratively in $ j = 1,2, \ldots, \ell $. Such a solution $ \mathbf{a} \in \mathbb{F}_q^{n_j} $ only exists for a unique value of $ j $, say $ i $ (at which point the iteration stops), and it is unique, since $ \mathcal{P} = \{ \mathcal{H}_1, \mathcal{H}_2, \ldots, \mathcal{H}_\ell \} $ is an improper partial $ \mathbf{n} $-spread. 

In conclusion, given a received vector $ \mathbf{y} = \mathbf{c} + \mathbf{e} \in \mathbb{F}_q^n $, where $ \mathbf{c} \in \mathcal{C} $ is the correct codeword and $ {\rm wt}_{SR}(\mathbf{e}) \leq 1 $, we proceed as follows:
\begin{enumerate}
\item
Compute $ \mathbf{s} = H \mathbf{y}^T \in \mathbb{F}_q^r $. If $ \mathbf{s} = \mathbf{0} $, then output $ \mathbf{c} = \mathbf{y} $ and exit the algorithm. 
\item
Until a solution is found, try to solve the linear system $ H_j \mathbf{a}^T = \mathbf{s} $ in the variables $ \mathbf{a} \in \mathbb{F}_q^{n_j} $, iteratively in $ j = 1,2, \ldots, \ell $.
\item
Given the unique solution $ (i,\mathbf{a}) $ from Step 2, output $ \mathbf{c} = \mathbf{y} - \mathbf{e}_{i,\mathbf{a}} $, where $ \mathbf{e}_{i,\mathbf{a}} $ is as in (\ref{eq def of coset leader for syndrome}).
\end{enumerate}

Step 1 has a complexity of $ \mathcal{O}(nr) $, whereas Step 2 has a complexity of $ \mathcal{O}(\ell r^3) $, both in number of operations in the field $ \mathbb{F}_q $. Hence, the total number of operations over $ \mathbb{F}_q $ is $ \mathcal{O}(nr + \ell r^3) $. In the case of proper sum-rank Hamming codes where $ N = n_1 = n_2 = \ldots = n_\ell $ divides $ r $, this complexity is
$$ \mathcal{O} \left( (r^3 + N r) \frac{q^r - 1}{q^N - 1} \right). $$
It is worth noting that the term $ (q^r - 1)(q^N - 1) $ is dominant over $ r^3 + Nr $ and thus the complexity is close to linear in the code length $ n $.

\section{Sum-rank simplex codes}  \label{sec sum-rank simplex codes}

In this section, we define sum-rank simplex codes as duals of sum-rank Hamming code, in analogy with the classical case (see \cite[Sec. 1.8]{pless}). Since we distinguished between proper and improper sum-rank Hamming codes, we do the same for sum-rank simplex codes. We consider dual codes with respect to the standard inner product in $ \mathbb{F}_{q^m}^n $.

\begin{definition} [\textbf{Sum-rank simplex codes}] \label{def sum-rank simplex codes}
A linear code $ \mathcal{C} \subseteq \mathbb{F}_{q^m}^n $ is called a proper (resp. improper) sum-rank simplex code if its dual $ \mathcal{C}^\perp \subseteq \mathbb{F}_{q^m}^n $ is a proper (resp. improper) sum-rank Hamming code.
\end{definition}

One may consider the case $ \ell = 1 $. The duals of the codes in Subsection \ref{subsec rank Hamming codes} are again of the same form (Cartesian products of Gabidulin codes), thus we do not investigate this case further.

In the rest of the section, we only consider again the case $ m = 1 $. We will also make use of partial spreads as in Subsection \ref{subsec partial spreads}. However, now our objective is to lower bound the minimum sum-rank distance of the corresponding proper sum-rank simplex code. 

The main result of this subsection is the following.

\begin{theorem} \label{th SR distance of proper simplex codes m=1 }
Let $ n = \ell N $ be a sum-rank length partition with equal sublengths $ N = n_1 = n_2 = \ldots = n_\ell $. Let $ r \geq N $ and let $ s \geq 0 $ be the remainder of $ r $ divided by $ N $. If $ \mathcal{C} \subseteq \mathbb{F}_q^n $ is a proper sum-rank simplex code of dimension $ r $, then it holds that
\begin{equation}
{\rm d}_{SR}(\mathcal{C}) \geq \left\lbrace \begin{array}{ll}
\left\lceil \frac{q^{r-1} (q-1)}{q^N - 1} \right\rceil & \textrm{ if } s = 0, \vspace*{0.5em} \\
\left\lfloor \frac{q^{r-1} (q-1)}{q^N - 1} \right\rfloor - q^s + 1 & \textrm{ if } s > 0.
\end{array} \right.
\label{eq SR distance of proper simplex codes m=1 }
\end{equation}
\end{theorem}
\begin{proof}
By Definition \ref{def sum-rank simplex codes} and Theorem \ref{th SR Hamming given by partial spreads}, $ \mathcal{C} $ has a generator matrix of the form
$$ H = (H_1, H_2, \ldots, H_\ell) \in \mathbb{F}_q^{r \times n}, $$
where $ H_i \in \mathbb{F}_q^{r \times N} $ and $ \mathcal{H}_i = {\rm Col}(H_i) \subseteq \mathbb{F}_q^r $, for $ i = 1,2, \ldots, \ell $, and $ \mathcal{P} = \{ \mathcal{H}_1, \mathcal{H}_2, \ldots, \mathcal{H}_\ell \} $ is a maximal-size partial $ N $-spread in $ \mathbb{F}_q^r $. 

Let $ \mathbf{x} \in \mathbb{F}_q^r $ be distinct from zero. Then it holds that
$$ {\rm wt}_{SR}(\mathbf{x} H) = \sum_{i=1}^\ell {\rm wt}_R(\mathbf{x} H_i) = \ell - | \{ i \in [\ell] \mid \mathcal{H}_i \subseteq \mathbf{x}^\perp \} |, $$
where $ \mathbf{x}^\perp \subseteq \mathbb{F}_q^r $ is the dual of the linear code generated by the vector $ \mathbf{x} $. In other words,
$$ {\rm wt}_{SR}(\mathbf{x} H) = \ell - | \mathcal{Q} |, $$
where $ \mathcal{Q} = \{ \mathcal{H}_i \mid \mathcal{H}_i \subseteq \mathbf{x}^\perp, 1 \leq i \leq \ell \} $ is the partial $ N $-spread $ \mathcal{P} $ restricted to the hyperplane $ \mathbf{x}^\perp $. 

Let $ s \geq 0 $ and $ t \geq 0 $ be the remainders of $ r $ and $ r-1 $ divided by $ N $, respectively. By Proposition \ref{prop size of maximal partial spread}, we deduce that
$$ {\rm wt}_{SR}(\mathbf{x}H) = \ell - |\mathcal{Q}| \geq \frac{q^r - q^s}{q^N - 1} - q^s + 1 - \frac{q^{r-1} - q^t}{q^N - 1}. $$
We will consider now the cases $ s = 0 $ and $ s > 0 $ separately.

Assume first that $ s = 0 $. Then $ t = N - 1 $ and
$$ \frac{q^r - q^s}{q^N - 1} - q^s + 1 - \frac{q^{r-1} - q^t}{q^N - 1} = \frac{q^r - 1}{q^N - 1} - \frac{q^{r-1} - q^{N-1}}{q^N - 1} \geq $$
$$ \frac{q^r - 1}{q^N - 1} - \frac{q^{r-1} - 1}{q^N - 1} = \frac{q^{r-1}(q-1)}{q^N - 1}, $$
and (\ref{eq SR distance of proper simplex codes m=1 }) is proven in this case.

Next assume that $ s > 0 $. Then $ t = s - 1 \geq 0 $ and
$$ \frac{q^r - q^s}{q^N - 1} - q^s + 1 - \frac{q^{r-1} - q^t}{q^N - 1} = \frac{q^r - q^{r-1}}{q^N - 1} - \frac{q^s - q^{s-1}}{q^N - 1} - q^s + 1. $$
Let $ h \geq 0 $ be the unique integer such that $ r = hN + s $. We have that
$$ \frac{q^r - q^{r-1}}{q^N - 1} = \frac{q^s - q^{s-1}}{q^N - 1} q^{hN} = (q^s - q^{s-1})\frac{q^{hN} - 1}{q^N - 1} + \frac{q^s - q^{s-1}}{q^N - 1}, $$
where the term $ (q^s - q^{s-1}) (q^{hN} - 1) / (q^N - 1) $ is an integer, and moreover
$$ 0 < \frac{q^s - q^{s-1}}{q^N - 1} < 1, $$
since $ s < N $. Thus we deduce that
$$ \left\lfloor \frac{q^r - q^{r-1}}{q^N - 1} \right\rfloor = (q^s - q^{s-1})\frac{q^{hN} - 1}{q^N - 1} = \frac{q^r - q^{r-1}}{q^N - 1} - \frac{q^s - q^{s-1}}{q^N - 1} . $$
Hence (\ref{eq SR distance of proper simplex codes m=1 }) is proven also in the case $ s > 0 $, and we are done.
\end{proof}

\begin{remark}
Observe that (\ref{eq SR distance of proper simplex codes m=1 }) is an equality in the case of classical simplex codes, which is recovered from Theorem \ref{th SR distance of proper simplex codes m=1 } by further setting $ N = 1 $, hence $ s = 0 $. Furthermore, a classical simplex code is a constant-weight code, meaning that any of its non-zero codewords has weight equal to $ q^{r-1} $. See \cite[Th. 1.8.3]{pless}. We leave the sharpness of (\ref{eq SR distance of proper simplex codes m=1 }) as an open problem, together with whether sum-rank simplex codes are in general constant-weight codes. Obviously, the use of sharper bounds on the sizes of maximal-size partial spreads than those in Proposition \ref{prop size of maximal partial spread} would help in this regard.
\end{remark}

\section{Some applications} \label{sec applications}

In this section, we provide two possible applications of sum-rank Hamming and simplex codes.

\subsection{Multishot matrix-multiplicative channels} \label{subsec matrix-mult channels}

Multishot matrix-multiplicative channels model, among others, multiple uses of linearly coded networks where the transmitter has no knowledge of the network or linear network code, and the linear combinations in intermediate nodes are not seen by the receiver. We refer to \cite{multishot} and \cite{secure-multishot} for details on multishot network coding. We note that a similar channel may be used for space-time coding, see \cite[Sec. III]{space-time-kumar}.

We follow \cite[Sec. II]{secure-multishot} to define multishot matrix-multiplicative channels with errors and erasures. Fix positive integers $ \ell , n_1, n_2, \ldots, n_\ell $, $ m_1, m_2, \ldots, m_\ell $ and denote $ \mathbf{n} = (n_1, n_2, \ldots, n_\ell) $ and $ \mathbf{m} = (m_1, m_2, \ldots, m_\ell) $. An $ (\mathbf{m} \times \mathbf{n}) $-multishot matrix-multiplicative channel with noise over the base field $ \mathbb{F}_q $ takes as input a sequence of matrices
$$ (X_1, X_2, \ldots, X_\ell) \in \mathbb{F}_q^{m_1 \times n_1} \times \mathbb{F}_q^{m_2 \times n_2} \times \cdots \times \mathbb{F}_q^{m_\ell \times n_\ell} $$
and outputs a sequence of matrices
\begin{equation}
(X_1A_1, X_2A_2, \ldots, X_\ell A_\ell) + (E_1, E_2, \ldots, E_\ell) \in \mathbb{F}_q^{m_1 \times N_1} \times \mathbb{F}_q^{m_2 \times N_2} \times \cdots \times \mathbb{F}_q^{m_\ell \times N_\ell},
\label{eq output matrix multiplicative channel}
\end{equation}
for some matrix $ A_i \in \mathbb{F}_q^{n_i \times N_i} $, called the $ i $th transfer matrix, and some matrix $ E_i \in \mathbb{F}_q^{m_i \times N_i} $, called the $ i $th error matrix, for $ i = 1,2, \ldots, \ell $, and for some positive integers $ N_1, N_2, \ldots, N_\ell $.

The number of sum-rank errors and erasures are given, under this model, by
\begin{equation}
\sum_{i=1}^\ell {\rm Rk}(E_i) \quad \textrm{and} \quad \sum_{i=1}^\ell (n_i - {\rm Rk}(A_i)),
\label{eq at most t and rho}
\end{equation}
respectively. As it was the case for the sum-rank metric recovering the Hamming and rank metrics, it holds that multishot matrix-multiplicative channels with errors and erasures as in (\ref{eq at most t and rho}) recover discrete memoryless channels with symbol-wise errors and erasures by setting $ m_1 = m_2 = \ldots = m_\ell = n_1 = n_2 = \ldots = n_\ell = N_1 = N_2 = \ldots = N_\ell = 1 $, and they recover singleshot matrix-multiplicative channels simply by setting $ \ell = 1 $.

We will only consider the case $ m = m_1 = m_2 = \ldots = m_\ell $. In such a case, we may consider a sequence of matrices in $ \mathbb{F}_q^{m \times n_1} \times \mathbb{F}_q^{m \times n_2} \times \cdots \times \mathbb{F}_q^{m \times n_\ell} $ as a vector in $ \mathbb{F}_{q^m}^n $, where $ n = n_1 + n_2 + \cdots + n_\ell $, via the matrix representation map in (\ref{eq def matrix representation map}). 

With such considerations, we say that a $ k $-dimensional linear code $ \mathcal{C} \subseteq \mathbb{F}_{q^m}^n $ can coherently correct $ t $ sum-rank errors and $ \rho $ sum-rank erasures if, for all transfer matrices $ A_i \in \mathbb{F}_q^{n_i \times N_i} $, for $ i = 1,2, \ldots, \ell $, such that $ \sum_{i=1}^\ell {\rm Rk}(A_i) \geq n - \rho $, there exists a decoder $ D : \mathbb{F}_{q^m}^N \longrightarrow \mathbb{F}_{q^m}^k $, possibly depending on $ A_1, A_2, \ldots, A_\ell $, where $ N = N_1 + N_2 + \cdots + N_\ell $, such that 
$$ D(\mathbf{c}^{(1)} A_1 + \mathbf{e}^{(1)}, \mathbf{c}^{(2)} A_2 + \mathbf{e}^{(2)}, \ldots, \mathbf{c}^{(\ell)} A_\ell + \mathbf{e}^{(\ell)}) = \mathbf{x}, $$
where $ \mathbf{c} = (\mathbf{c}^{(1)}, \mathbf{c}^{(2)}, \ldots, \mathbf{c}^{(\ell)}) \in \mathbb{F}_{q^m}^n $ is the encoding of $ \mathbf{x} \in \mathbb{F}_{q^m}^n $ by some fixed generator matrix of $ \mathcal{C} $, for all $ \mathbf{x} \in \mathbb{F}_{q^m}^k $ and all $ \mathbf{e} = (\mathbf{e}^{(1)}, \mathbf{e}^{(2)}, \ldots, \mathbf{e}^{(\ell)}) \in \mathbb{F}_{q^m}^N $ such that
$$ {\rm wt}_{SR}(\mathbf{e}^{(1)}, \mathbf{e}^{(2)}, \ldots, \mathbf{e}^{(\ell)}) = \sum_{i=1}^\ell {\rm Rk}(E_i) \leq t, $$
where $ E_i = M_\mathcal{A}(\mathbf{e}^{(i)}) \in \mathbb{F}_q^{m \times N_i} $, for $ i = 1,2, \ldots, \ell $. The word coherently refers to the fact that the transfer matrices $ A_1, A_2, \ldots, A_\ell $ are known to the decoder (that is why it may depend on such matrices). An adaptation to the non-coherent case may be done by lifting, see \cite[Subsec. IV-C]{secure-multishot}.

With these definitions, the following result is a particular case of \cite[Th. 1]{secure-multishot}.

\begin{proposition} [\textbf{\cite{secure-multishot}}]
Let $ n = n_1 + n_2 + \cdots + n_\ell $ be a sum-rank length partition, and let $ t \geq 0 $ and $ \rho \geq 0 $ be integers. A linear code $ \mathcal{C} \subseteq \mathbb{F}_{q^m}^n $ can coherently correct up to $ t $ sum-rank errors and $ \rho $ sum-rank erasures in any $ (\mathbf{m} \times \mathbf{n}) $-multishot matrix-multiplicative channel if, and only if, it satisfies that $ {\rm d}_{SR}(\mathcal{C}) > 2t + \rho $.
\end{proposition}

Hence, sum-rank Hamming codes are the longest linear codes that can correct a single sum-rank error ($ t = 1 $ and $ \rho = 0 $) over such channels. 

\begin{theorem}
Let $ n = \ell N $ be a sum-rank length partition with equal sublengths $ N = n_1 = n_2 = \ldots = n_\ell $. A proper sum-rank Hamming code $ \mathcal{C} \subseteq \mathbb{F}_{q^m}^n $ with redundancy $ r = n - \dim(\mathcal{C}) $ is the linear code of redundancy $ r $ capable of correcting a single sum-rank error in an $ (\mathbf{m} \times \mathbf{n}) $-multishot matrix-multiplicative channel with $ m_1 = m_2 = \ldots = m_\ell $ and $ n_1 = n_2 = \ldots = n_\ell $ for the largest number of shots $ \ell $.
\end{theorem} 

Improper sum-rank Hamming code admit a similar interpretation for multishot matrix-multiplicative channels. However, their main feature is not necessarily admitting the largest number of shots $ \ell $, but having the largest overall sum-rank length $ n $ attainable by increasing $ \ell $ and/or the numbers of matrix columns $ n_1, n_2, \ldots, n_\ell $.

We conclude by noting that, if $ m = 1 $, then no code can correct a single sum-rank error if $ \ell \leq 2 $, since we have that $ {\rm d}_{SR}(\mathcal{C}) \leq \ell $ when $ m = 1 $. In particular, no rank error-correcting code exists if $ m = 1 $, as in that case $ \ell = 1 $. Therefore, in this work we have provided the first known error-correcting codes for matrix-multiplicative channels with $ m = 1 $ but $ N = n_1 = n_2 = \ldots = n_\ell > 1 $, by making use of several shots of the channel, i.e. $ \ell \geq 3 $. In particular, choosing $ q = 2 $ and $ m = 1 $, we have provided the first codes that can correct bit-wise errors in linearly coded networks where the transmitter has no knowledge of the network or linear network code and the linear combinations of intermediate nodes are not known by the receiver. The decoding algorithm from Subsection \ref{subsec syndrome decoding m=1 } can be trivially implemented in this scenario.

\subsection{Locally repairable codes} \label{subsec lrc}

Locally repairable codes \cite{gopalan} are an attractive alternative to MDS codes for large distributed storage systems, since they allow to repair a single erasure (most common erasure pattern) by contacting a small number, called locality, of other nodes, while being able to correct more erasures in catastrophic cases. 

In \cite{lrc-sum-rank}, we established a connection between the sum-rank metric and locally repairable codes. We start by recalling the following result, which is \cite[Cor. 4]{lrc-sum-rank}.

\begin{lemma} [\textbf{\cite{lrc-sum-rank}}] \label{lemma global code for lrc}
Let $ n = n_1 + n_2 + \cdots + n_\ell $ and let $ \mathcal{C}_{out} \subseteq \mathbb{F}_{q^m}^n $ be a linear code. Choose an $ n_i $-dimensional local linear code $ \mathcal{C}_i \subseteq \mathbb{F}_q^{N_i} $ with generator matrix $ A_i \in \mathbb{F}_q^{n_i \times N_i} $ (thus $ n_i \leq N_i $), for $ i = 1,2, \ldots, \ell $. Define the global code $ \mathcal{C}_{glob} \subseteq \mathbb{F}_{q^m}^M $ with total length $ M = N_1 + N_2 + \cdots + N_\ell $ as
\begin{equation}
\mathcal{C}_{glob} = \mathcal{C}_{out} {\rm diag}(A_1, A_2, \ldots, A_\ell) \subseteq \mathbb{F}_{q^m}^M .
\label{eq def global code for lrc}
\end{equation}
Denote by $ \Gamma_i \subseteq [M] $ the set of coordinates ranging from $ \sum_{j=1}^{i-1} N_j + 1 $ to $ \sum_{j=1}^i N_j $, for $ i = 1,2, \ldots, \ell $. Then $ \dim(\mathcal{C}_{glob}) = \dim(\mathcal{C}_{out}) $ and any erasure pattern $ \mathcal{E}_i \subseteq \Gamma_i $ with $ | \mathcal{E}_i | < {\rm d}_H(\mathcal{C}_i) $ can be corrected by $ \mathcal{C}_{glob} $ with the same complexity over the same field as with $ \mathcal{C}_i $ and only using the $ N_i - |\mathcal{E}_i| $ symbols in $ \Gamma_i \setminus \mathcal{E}_i $, for $ i = 1,2, \ldots, \ell $. 
\end{lemma}

In other words, the global code $ \mathcal{C}_{glob} $ keeps the same dimension as the outer code $ \mathcal{C}_{out} $, but it also has the local erasure-correction capability of the local codes $ \mathcal{C}_1, \mathcal{C}_2, \ldots, \mathcal{C}_\ell $. Next, we recall the global erasure-correction capability of the global code in terms of the sum-rank erasure-correction capability of the outer code. The following result follows from Theorem \ref{th sum-rank distance is min among hamming distances} and it gives a simple sufficient condition for a global erasure pattern to be correctable by the global code.

\begin{proposition} \label{prop global erasure corrections}
Let notation be as in Lemma \ref{lemma global code for lrc}. Let $ \mathcal{E} \subseteq [M] $ be an erasure pattern and define $ \mathcal{E}_i = \mathcal{E} \cap \Gamma_i $ and $ \mathcal{R}_i = \Gamma_i \setminus \mathcal{E}_i $, for $ i = 1,2, \ldots, \ell $. Then the erasure pattern $ \mathcal{E} $ can be corrected by $ \mathcal{C}_{glob} $ if 
$$ {\rm d}_{SR}(\mathcal{C}_{out}) > n - \sum_{i=1}^\ell {\rm Rk}(A_i \vert_{\mathcal{R}_i}), $$
where $ {\rm d}_{SR} $ is considered with respect to the sum-rank length partition $ n = n_1 + n_2 + \cdots + n_\ell $.
\end{proposition}

Observe that the local codes could in principle help in the global erasure correction by repairing erasures locally (inside each $ \Gamma_i $) whenever they can, as in Lemma \ref{lemma global code for lrc}. However, they would only be adding redundant symbols, hence they would not increase the ranks of the matrices $ A_i \vert_{\mathcal{R}_i} $, thus such a local erasure correction does not affect the global erasure-correction capability, as shown in Proposition \ref{prop global erasure corrections}.

Usually, the local codes are be chosen to be MDS, as they must be short codes. Moreover, they are usually considered such that 
\begin{equation}
A_i = \left( \begin{array}{ccccc}
1 & 0 & \ldots & 0 & 1 \\
0 & 1 & \ldots & 0 & 1 \\
\vdots & \vdots & \ddots & \vdots & \vdots \\
0 & 0 & \ldots & 1 & 1 \\
\end{array} \right) \in \mathbb{F}_q^{n_i \times (n_i + 1)},
\label{eq generator matrix for local MDS codes}
\end{equation}
preferably with $ q = 2 $. In that case, all local codes may correct up to one erasure very efficiently with only one XOR operation. The global code $ \mathcal{C}_{glob} $ would possibly have unequal localities $ n_1, n_2, \ldots, n_\ell $ (locality $ n_i $ for the $ i $th local group $ \Gamma_i $). However, the construction in (\ref{eq def global code for lrc}) allows any choice of local linear codes over $ \mathbb{F}_q $.

In the case of equal localities $ N = n_1 = n_2 = \ldots = n_\ell $, we may provide locally repairable codes over any field $ \mathbb{F}_q $, including $ q = 2 $, by making use of proper sum-rank Hamming codes with $ m = 1 $. In that case, we deduce the following construction from Corollary \ref{cor length of proper sum-rank hamming m=1 }.

\begin{theorem} \label{th lrcs from sum-rank hamming codes}
Let $ q $ be an arbitrary prime power (e.g. $ q = 2 $) and let $ n = \ell N $ be a sum-rank length partition with equal sublengths $ N = n_1 = n_2 = \ldots = n_\ell $. Let $ r \geq N $ be such that $ N $ divides $ r $. Let $ \mathcal{C}_{out} \subseteq \mathbb{F}_q^n $ be a proper sum-rank Hamming code with redundancy $ r = n - \dim(\mathcal{C}_{out}) $ and let the $ i $th local code $ \mathcal{C}_i \subseteq \mathbb{F}_q^{N_i} $ have a generator matrix as in (\ref{eq generator matrix for local MDS codes}). Then the global code given as in (\ref{eq def global code for lrc}) has locality $ N $ and number of local groups, total length and dimension given by
$$ \ell = \frac{q^r - 1}{q^N - 1}, \quad M = (N+1) \frac{q^r - 1}{q^N - 1} \quad \textrm{and} \quad k = N \frac{q^r - 1}{q^N - 1} - r = N \left( \frac{q^r - 1}{q^N - 1} - h \right) , $$
respectively, where $ r = h N $. Therefore, according to \cite[Def. 6]{lrc-sum-rank}, the number of local parities is $ 1 $ per local group (i.e., the redundancy of each local code) and the number of global parities is $ r $ (i.e., the redundancy of the outer code). Finally, $ \mathcal{C}_{glob} $ may correct one erasure locally by a simple XOR of the remaining nodes in that local group, and it may correct any erasure pattern with one erasure per local group plus any other $ 2 $ extra erasures anywhere.
\end{theorem}

Table \ref{table parameters lrc} gives some choices of parameters for linear codes as in Theorem \ref{th lrcs from sum-rank hamming codes} for $ q = 2 $.

\begin{table}[t]
\centering
\begin{tabular}{c||c|c||c|c||c|c||c|c}
Locality $ N $ & 2 & 2 & 3 & 3 & 4 & 4 & 5 & 5 \\
\hline
No. local groups $ \ell $ & 4 & 21 & 9 & 73 & 17 & 273 & 33 & 1057 \\
\hline
Global parities $ r $ & 4 & 6 & 6 & 9 & 8 & 12 & 10 & 15 \\
\hline
Dimension $ k $ & 6 & 36 & 21 & 210 & 60 & 1080 & 155 & 5270 \\
\hline
Length $ M $ & 12 & 63 & 36 & 292 & 85 & 1365 & 198 & 6342 \\
\end{tabular} \\
\caption{Some parameters for the linear locally repairable codes described in Theorem \ref{th lrcs from sum-rank hamming codes} for $ q = 2 $, which may correct any $ 2 $ extra erasures on top of $ 1 $ erasure per local group.}
\label{table parameters lrc}
\end{table}

We have chosen to provide a construction based on sum-rank Hamming codes with $ m = 1 $, since it allows to work on any finite field $ \mathbb{F}_q $ (in particular $ q = 2 $) and since simplex codes have low information rate. Having high information rate even though only up to $ 2 $ extra erasures may be corrected is a desirable parameter regime for distributed storage applications.

Finally, note that we could have chosen any family of local linear codes in Theorem \ref{th lrcs from sum-rank hamming codes}. This would make sense, for instance, to obtain hierachical locally repairable codes by choosing local codes that are in turn locally repairable codes (see \cite[Subsec. V-C]{lrc-sum-rank}).

\section*{Acknowledgement}

The author wishes to thank Frank R. Kschischang for valuable discussions on this manuscript. The author also gratefully acknowledges the support from The Independent Research Fund Denmark (Grant No. DFF-7027-00053B).

\footnotesize
 

\end{document}